\documentclass[a4paper,11pt]{article}
\usepackage[utf8]{inputenc}
\pdfoutput=1 
\usepackage[margin=3cm]{geometry}
\usepackage{amsmath}
\usepackage{amssymb}
\usepackage{amsthm}
\usepackage{float}
\usepackage{authblk}

\usepackage{microtype}
\usepackage{hyperref}
\usepackage[svgnames]{xcolor}
\hypersetup{colorlinks={true},urlcolor={blue},linkcolor={DarkBlue},citecolor={DarkGreen}}
\usepackage[capitalise,nameinlink]{cleveref}

\usepackage[numbers]{natbib}
\usepackage{doi}

\usepackage{bm}
\usepackage{bold-extra}

\theoremstyle{definition}
\newtheorem{definition}{Definition}

\theoremstyle{plain}
\newtheorem{theorem}{Theorem}
\newtheorem{lemma}{Lemma}
\newtheorem{corollary}{Corollary}
\newtheorem{proposition}{Proposition}

\theoremstyle{remark}
\newtheorem{remark}{Remark}
\newtheorem*{remark*}{Remark}

\newtheorem*{claim*}{Claim}

\newcommand{\set}{\textsf{Set}}

\newcommand{\ppad}{\textup{PPAD}}
\newcommand{\ppa}{\textup{PPA}}
\newcommand{\degtag}[1]{[\# #1]}
\newcommand{\ppak}[1]{\textup{PPA-$#1$}}
\newcommand{\ppakl}[2]{\ppak{#1\degtag{#2}}}
\newcommand{\ppmod}[1]{\textup{$\text{PMOD}^{#1}$}}
\newcommand{\modd}[1]{$\textup{MOD}^{#1}$}
\newcommand{\bipa}[1]{\textsc{Bipartite-mod-$#1$}}
\newcommand{\bipal}[2]{\bipa{#1\degtag{#2}}}
\newcommand{\imba}[1]{\textsc{Imbalance-mod-$#1$}}
\newcommand{\imbal}[2]{\imba{#1\degtag{#2}}}
\newcommand{\hyper}[1]{\textsc{Hypergraph-mod-$#1$}}
\newcommand{\hyperl}[2]{\hyper{#1\degtag{#2}}}
\newcommand{\group}[1]{\textsc{Partition-mod-$#1$}}
\newcommand{\groupl}[2]{\group{#1\degtag{#2}}}
\newcommand{\pfactors}[1]{\operatorname{\mathsf{PF}}(#1)}
\DeclareMathOperator*{\amper}{\&}

\begin{document}

\title{The Classes PPA-$k$:\\ Existence from Arguments Modulo $k$}

\author{Alexandros Hollender\thanks{A preliminary version of this paper appeared in the Proceedings of the 15th Conference on Web and Internet Economics (WINE 2019).}}

\affil{Department of Computer Science, University of Oxford\\ {\tt alexandros.hollender@cs.ox.ac.uk}}

\date{}

\maketitle

\begin{abstract}
The complexity classes PPA-$k$, $k \geq 2$, have recently emerged as the main candidates for capturing the complexity of important problems in fair division, in particular Alon's \textsc{Necklace-Splitting} problem with $k$ thieves. Indeed, the problem with two thieves has been shown complete for PPA $=$ PPA-$2$. In this work, we present structural results which provide a solid foundation for the further study of these classes. Namely, we investigate the classes PPA-$k$ in terms of (i) equivalent definitions, (ii) inner structure, (iii) relationship to each other and to other TFNP classes, and (iv) closure under Turing reductions.
\end{abstract}

\section{Introduction}

The complexity class TFNP is the class of all search problems such that every instance has a least one solution and any solution can be checked in polynomial time. It has attracted a lot of interest, because, in some sense, it lies between P and NP. Moreover, TFNP contains many natural problems for which no polynomial-time algorithm is known, such as \textsc{Factoring} (given a integer, find a prime factor) or \textsc{Nash} (given a bimatrix game, find a Nash equilibrium). However, no problem in TFNP can be NP-hard, unless NP $=$ co-NP~\cite{megiddo1991total}. Furthermore, it is believed that no TFNP-complete problem exists~\cite{papadimitriou1994complexity,pudlak2015complexity}. Thus, the challenge is to find some way to provide evidence that these TFNP problems are indeed hard.

Papadimitriou~\cite{papadimitriou1994complexity} proposed the following idea: define \emph{subclasses} of TFNP and classify the natural problems of interest with respect to these classes. Proving that many natural problems are complete for such a class, shows that they are ``equally'' hard. Then, investigating how these classes relate to each other, yields a relative classification of all these problems. In other words, it provides a unified framework that gives a better understanding of how these problems relate to each other. TFNP subclasses are based on various non-constructive existence results. Some of these classes and their corresponding existence principle are:
\begin{itemize}
    \item PPAD : given a directed graph and an unbalanced vertex (i.e., out-degree $\neq$ in-degree), there must exist another unbalanced vertex.
    \item PPA : given an undirected graph and vertex with odd degree, there must exist another vertex with odd degree (Handshaking Lemma).
    \item PPP : given a function mapping a finite set to a smaller set, there must exist a collision (Pigeonhole Principle).
\end{itemize}
Other TFNP subclasses are PPADS, PLS~\cite{johnson1988easy}, CLS~\cite{daskalakis2011continuous}, PTFNP~\cite{goldberg2018towards}, EOPL and UEOPL \cite{fearnley2018unique}. It is known that $\ppad \subseteq \textup{PPADS} \subseteq \textup{PPP}$, $\ppad \subseteq \ppa$ and $\textup{UEOPL} \subseteq \textup{EOPL} \subseteq \textup{CLS} \subseteq \ppad \cap \textup{PLS}$. Very recently it was shown that in fact $\textup{CLS} = \ppad \cap \textup{PLS}$~\cite{FGHS21-GD}. Any separation between TFNP subclasses would imply P $\neq$ NP, but various oracle separations exist~\cite{beame1998relative,morioka2001classification,buresh2004relativized,buss2012propositional} (see \cref{sec:prelims} for more details).

TFNP subclasses have been very successful in capturing the complexity of natural problems. The most famous result is that the problem \textsc{Nash} is PPAD-complete~\cite{daskalakis2009complexity,CDT}, but various other natural problems have also been shown PPAD-complete~\cite{CSVY08,CDDT09,CPY17,KPRST13}. Many local optimisation problems have been proved PLS-complete~\cite{johnson1988easy,papadimitriou1990local,krentel1989structure,fabrikant2004pure,dumrauf2009multiprocessor}. Recently, the first natural complete problems were found for PPA~\cite{FRG18-Consensus,FRG19-Necklace} and PPP~\cite{SZZ2018pppcomplete}. The famous \textsc{Factoring} problem has been partially related to PPA and PPP~\cite{J-JCSS16}.

\paragraph*{Necklace-Splitting.}
The natural problem recently shown PPA-complete is a problem in fair division, called the \textsc{$2$-Necklace-Splitting} problem~\cite{FRG19-Necklace}. For $k \geq 2$, the premise of the \textsc{$k$-Necklace-Splitting} problem is as follows. Imagine that $k$ thieves have stolen a necklace that has beads of different colours. Since the thieves are unsure of the value of the different beads, they want to divide the necklace into $k$ parts such that each part contains the same number of beads of each colour. However, the string of the necklace is made of precious metal, so the thieves don't want to use too many cuts. Alon's famous result~\cite{alon1987necklace} says that this can always be achieved with a limited number of cuts.

The corresponding computational problem can be described as follows. We are given an open necklace (i.e., a segment) with $n$ beads of $c$ different colours, i.e., there are $a_i$ beads of colour $i$ and $\sum_{i=1}^c a_i = n$. Furthermore, assume that for each $i$, $a_i$ is divisible by $k$ (the number of thieves). The goal is to cut the necklace in (at most) $c(k-1)$ places and allocate the pieces to the $k$ thieves, such that every thief gets exactly $a_i/k$ beads of colour $i$, for each colour $i$. By Alon's result~\cite{alon1987necklace}, a solution always exists, and thus the problem lies in TFNP.

The complexity of this problem has been an open problem for almost 30 years~\cite{papadimitriou1994complexity}. While the 2-thieves version is now resolved, the complexity of the problem with $k$ thieves ($k \geq 3$) remains open. The main motivation of the present paper is to investigate the classes \ppak{k}, which are believed to be the most likely candidates to capture the complexity of \textsc{$k$-Necklace-Splitting}. Indeed, in the conclusion of the paper where they prove that \textsc{$2$-Necklace-Splitting} is PPA-complete, Filos-Ratsikas and Goldberg~\cite[arXiv version]{FRG19-Necklace} mention:

\begin{quotation}
``What is the computational complexity of $k$-thief \textsc{Necklace-splitting}, for $k$ not a power of 2? As discussed in \cite{meunier2014simplotopal,dLZ2006borsuk}, the proof that it is a total search problem, does \emph{not} seem to boil down to the PPA principle. Right now, we do not even know if it belongs to PTFNP~\cite{goldberg2018towards}.

Interestingly, Papadimitriou in \cite{papadimitriou1994complexity} (implicitly) also defined a number of computational complexity classes related to PPA, namely PPA-$p$, for a parameter $p \geq 2$. [...] Given the discussion above, it could possibly be the case that the principle associated with \textsc{Necklace-Splitting} for $k$-thieves is the PPA-$k$ principle instead.''
\end{quotation}

\paragraph*{$\bm{\ppak{p}}$.}
The TFNP subclasses \ppak{p} were defined by Papadimitriou almost 30 years ago in his seminal paper \cite{papadimitriou1994complexity}. Recall that the existence of a solution to a PPA problem is guaranteed by a parity argument, i.e., an argument modulo $2$. The classes \ppak{p} are a generalisation of this. For every prime $p$, the existence of a solution to a \ppak{p} problem is guaranteed by an argument modulo $p$. In particular, $\ppak{2} = \ppa$. Surprisingly, these classes have received very little attention. As far as we know, they have only been studied in the following:
\begin{itemize}
\item Papadimitriou~\cite{papadimitriou1994complexity} defined the classes \ppak{p} and proved that a problem called \textsc{Chevalley-mod-$p$} lies in \ppak{p} and a problem called \textsc{Cubic-Subgraph} lies in \ppak{3}.
\item In an online thread on Stack Exchange~\cite{jerabekstackexchange}, Je{\v r}{\'a}bek provided two other equivalent ways to define \ppak{3}. The problems and proofs can be generalised to any prime $p$.
\item In his thesis~\cite{johnson2011thesisNPsearch}, Johnson defined the classes \ppmod{k} for any $k \geq 2$, which were intended to capture the complexity of counting arguments modulo $k$. He proved various oracle separation results involving his classes and other TFNP classes. While the \ppak{p} classes are not mentioned by Johnson, using Je{\v r}{\'a}bek's results~\cite{jerabekstackexchange} it is easy to show that $\ppmod{p} = \ppak{p}$ for any prime $p$. In \cref{sec:pmod}, we characterise \ppmod{k} in terms of the classes \ppak{p} when $k$ is not prime. In particular, we show that \ppmod{k} only partially captures existence arguments modulo $k$.
\end{itemize}

\paragraph*{Our contribution.}
In this paper, we use the natural generalisation of Papadimitriou's definition of the classes \ppak{p} to define \ppak{k} for any $k \geq 2$. We then provide a characterisation of \ppak{k} in terms of the classes \ppak{p}. In particular, we show that \ppak{k} is completely determined by the set of prime factors of $k$. In order to gain a better understanding of the inner structure of the class \ppak{k}, we also define new subclasses that we denote \ppakl{k}{\ell} and investigate how they relate to the other classes. We show that \ppakl{k}{\ell} is completely determined by the set of prime factors of $k/\gcd(k,\ell)$.

Furthermore, we provide various equivalent complete problems that can be used to define \ppak{k} and \ppakl{k}{\ell} (\cref{sec:complete-toolbox}). While these problems are not ``natural'', we believe that they provide additional tools that can be very useful when proving that natural problems are complete for these classes. In \cref{sec:turing-reduction}, we provide an additional tool for showing that problems lie in these classes: we prove that \ppak{p^r} ($p$ prime, $r \geq 1$) and \ppakl{k}{\ell} ($k \geq 2$) are closed under Turing reductions. On the other hand, we provide evidence that \ppak{k} might not be closed under Turing reductions when $k$ is not a prime power.

Finally, in \cref{sec:pmod} we investigate the classes \ppmod{k} defined by Johnson~\cite{johnson2011thesisNPsearch} and provide a full characterisation in terms of the classes \ppak{k}. In particular, we show that $\ppmod{k} = \ppak{k}$ if $k$ is a prime power. However, when $k$ is not a prime power, we provide evidence that \ppmod{k} does not capture the full strength of existence arguments modulo $k$, unlike \ppak{k}. This characterisation of \ppmod{k} in terms of \ppak{k} leads to some oracle separation results involving \ppak{k} and other TFNP classes (using Johnson's oracle separation results).

We note that a significant fraction of our results were also obtained by G\"o\"os, Kamath, Sotiraki and Zampetakis in concurrent and independent work~\cite{GKSZ19-modulo-q}. In their work, they have also provided the first ``natural'' complete problem for the classes \ppak{p} (a variant of \textsc{Chevalley-mod-$p$}), namely the first complete problem that does not involve circuits or other computational devices in its description. The present work, and in particular the equivalent characterisations of the classes \ppak{k}, have been pivotal in subsequent work~\cite{FRHSZ21necklace} showing that the \textsc{$k$-Necklace-Splitting} problem lies in \ppak{k} under Turing reductions. However, the question of whether \textsc{$k$-Necklace-Splitting} is also \ppak{k}-hard remains open.

\section{Preliminaries}\label{sec:prelims}

\paragraph*{TFNP.}
Let $\{0,1\}^*$ denote the set of all finite length bit-strings and for $w \in \{0,1\}^*$ let $|w|$ be its length. A computational search problem is given by a binary relation $R \subseteq \{0,1\}^* \times \{0,1\}^*$. The problem is: given an instance $I \in \{0,1\}^*$, find an $s \in \{0,1\}^*$ such that $(I,s) \in R$, or return that no such $s$ exists. The search problem $R$ is in FNP (\emph{Functions in NP}), if $R$ is polynomial-time computable (i.e., $(I,s) \in R$ can be decided in polynomial time in $|I|+|s|$) and there exists some polynomial $p$ such that $(I,s) \in R \implies |s| \leq p(|I|)$.  Thus, FNP is the search problem version of NP (and FNP-complete problems are equivalent to NP-complete problems under Turing reductions).

The class TFNP (\emph{Total Functions in NP}~\cite{megiddo1991total}) contains all FNP search problems $R$ that are \emph{total}: for every $I \in\{0,1\}^*$ there exists $s \in\{0,1\}^*$ such that $(I,s) \in R$. With a slight abuse of notation, we can say that P lies in TFNP. Indeed, if a decision problem is solvable in polynomial time, then both the ``yes'' and ``no'' answers can be verified in polynomial time. In this sense, TFNP lies between P and NP.

Note that TFNP problems are not promise problems, i.e., we are not allowed to restrict the instance space $\{0,1\}^*$. This means that for any instance in $\{0,1\}^*$, there must always exist at least one solution. Nevertheless, TFNP can indirectly capture various settings where the instance space is restricted. For example, if a problem $R$ in FNP is total only on a subset $L$ of the instances and $L \in P$, then we can transform it into an equivalent TFNP problem by adding $(I,0)$ to $R$ for all instances $I \notin L$.

\paragraph*{Reductions.} Let $R$ and $S$ be total search problems in TFNP. We say that $R$ (many-one) reduces to $S$, denoted $R \leq S$, if there exist polynomial-time computable functions $f,g$ such that
$$(f(I),s) \in S \implies (I,g(I,s)) \in R.$$
Note that if $S$ is polynomial-time solvable, then so is $R$. We say that two problems $R$ and $S$ are (polynomial-time) equivalent, if $R \leq S$ and $S \leq R$.

There is also a more general type of reduction. A Turing reduction from $R$ to $S$ is a polynomial-time oracle Turing machine that solves problem $R$ with the help of queries to an oracle for $S$. Note that a Turing reduction that only makes a single oracle query immediately yields a many-one reduction.

\paragraph*{Encoding of Sets.} Many of the computational problems we consider in this paper involve Boolean circuits whose output is interpreted as a set. For example, it will often be the case that a circuit $C$ takes as input a bit-string in $\{0,1\}^n$ and outputs a set of at most $m$ bit-strings in $\{0,1\}^n$. We will denote this by $C: \{0,1\}^n \to \set_{\leq m}(\{0,1\}^n)$. Of course, a Boolean circuit has a fixed number of output bits and so the circuit will in fact be of the form $C: \{0,1\}^n \to \{0,1\}^t$, for some $t$ that is sufficiently large so that there are enough bits to encode any set of size at most $m$. It is easy to see that taking $t = mn + 1$ is enough. Indeed, we can for example use the following encoding: the set $\{x_1,\dots,x_\ell\} \subseteq \{0,1\}^n$, $\ell \leq m$, is represented by the bit-string $x_1 \cdots x_\ell 1 0^{(m-\ell)n} \in \{0,1\}^{mn+1}$. Clearly, we can efficiently check whether a bit-string in $\{0,1\}^{mn+1}$ is a valid representation of a set, and if not, we can just interpret it as the empty set $\emptyset \subset \{0,1\}^n$.

\paragraph*{PPA.} The class PPA (Polynomial Parity Argument)~\cite{papadimitriou1994complexity} is defined as the set of all TFNP problems that many-one reduce to the problem \textsc{Leaf}~\cite{papadimitriou1994complexity,beame1998relative}: given an undirected graph with maximum degree 2 and a leaf (i.e., a vertex of degree 1), find another leaf. The important thing to note is that the graph is not given explicitly (in which case the problem would be very easy), but it is provided implicitly through a succinct representation.

The vertex set is $\{0,1\}^n$ and the undirected graph is represented by a Boolean circuit $C: \{0,1\}^n \to \set_{\leq 2}(\{0,1\}^n)$. By this we mean that for any $x \in \{0,1\}^n$, we interpret $C(x)$ as the set of potential neighbours of $x$, where we syntactically enforce that $x \notin C(x)$. We say that there is an edge between $x$ and $y$ if $x \in C(y)$ and $y \in C(x)$. Thus, every vertex has at most two neighbours. Note that the size of the graph can be exponential with respect to its description size.

The full formal definition of the problem \textsc{Leaf} is: given a Boolean circuit $C: \{0,1\}^n \to \set_{\leq 2}(\{0,1\}^n)$ representing an undirected graph on the vertex set $\{0,1\}^n$ such that $|C(0^n)|=1$ (i.e., $0^n$ is a leaf), find
\begin{itemize}
    \item $x \neq 0^n$ such that $|C(x)| = 1$ (another leaf)
    \item or $x,y$ such that $x \in C(y)$ but $y \notin C(x)$ (an inconsistent edge)
\end{itemize}

\paragraph*{Type 2 Problems and Oracle Separations.} We work in the standard Turing machine model, but TFNP subclasses have also been studied in the black-box model. In this model, one considers the type 2 versions of the problems, namely, the circuits in the input are replaced by black-boxes. In that case, it is possible to prove unconditional separations between type 2 TFNP subclasses (in the standard model this would imply P $\neq$ NP). The interesting point here is that separations between type 2 classes yield separations of the corresponding classes in the standard model with respect to any generic oracle (see~\cite{beame1998relative} for more details on this). This technique has been used to prove various oracle separations between TFNP subclasses~\cite{beame1998relative,morioka2001classification,buresh2004relativized,buss2012propositional}. In \cref{sec:pmod} we provide some oracle separations involving \ppak{k} and other TFNP subclasses.

On the other hand, any reduction that works in the type 2 setting, also works in the standard setting. Indeed, it suffices to replace the calls to the black boxes by the corresponding circuits that compute them. In this paper, our reductions are stated in the standard model, but they also work in the type 2 setting, because they don't examine the inner workings of the circuits.

\section{Definition of the Classes}

\subsection{\texorpdfstring{$\bm{\ppak{k}}$}{PPA-k} : Polynomial Argument modulo \texorpdfstring{$\bm{k}$}{k}}

For any prime $p$, Papadimitriou~\cite{papadimitriou1994complexity} defined the class \ppak{p} as the set of all TFNP problems that many-one reduce to the following problem, that we call \bipa{p}: We are given an undirected bipartite graph (implicitly represented by a circuit) and a vertex with degree $\neq 0 \mod p$ (which we call the \emph{trivial solution}). The goal is to find another such vertex. This problem lies in TFNP: if all other vertices had degree $= 0 \mod p$, then the sum of the degrees of all vertices on each side would have a different value modulo $p$, which is impossible.

The problem remains well-defined and total if $p$ is not a prime, and so we will instead define it for any $k \geq 2$. Let us now provide a formal definition of the problem. A vertex of the bipartite graph is represented as a bit-string in $\{0,1\} \times \{0,1\}^n$, where the first bit indicates whether the vertex lies on the ``left'' or ``right'' side of the bipartite graph. The graph will be represented by a Boolean circuit that outputs a set of potential neighbours, just as we did for \textsc{Leaf}. Instead of at most two neighbours, here we allow at most $k$ neighbours (see \cref{rem:degin0k} for why this is enough). Note that we can syntactically enforce that the graph is bipartite, i.e., that a vertex $0x$ can only have neighbours of the type $1y$ and vice versa.

\begin{definition}[\bipa{k} \cite{papadimitriou1994complexity}]\label{def:bipartite-k}
Let $k \geq 2$. The problem \bipa{k} is defined as: given a Boolean circuit $C: \{0,1\} \times \{0,1\}^n \to \set_{\leq k}(\{0,1\} \times \{0,1\}^n)$ representing a bipartite graph on the vertex set $(\{0\} \times \{0,1\}^n, \{1\} \times \{0,1\}^n)$ with $|C(00^n)| \in \{1,\dots, k-1\}$, find
\begin{itemize}
    \item $x \neq 00^n$ such that $|C(x)| \notin \{0,k\}$
    \item or $x,y$ such that $y \in C(x)$ but $x \notin C(y)$.
\end{itemize}
\end{definition}
Here the trivial solution is the vertex $00^n$. The first type of solution corresponds to a vertex with degree $\neq 0 \mod k$. The second type of solution corresponds to an edge that is not well-defined. We can always ensure that all edges are well-defined by doing some pre-processing. Indeed, in polynomial time we can construct a circuit $C'$ such that all solutions are of the first type and yield a solution for $C$. On input $0x$ the circuit $C'$ first computes $C(0x) = \{1y_1, \dots, 1y_m\}$ and then for each $i$ removes $1y_i$ from this set, if $0x \notin C(1y_i)$.

\begin{remark}\label{rem:degin0k}
Note that in this problem statement we require that all degrees lie in $\{0,1,\dots, k\}$. This is easily seen to be equivalent to the more general formulation where vertices can have more than $k$ neighbours. Indeed, any vertex that has more than $k$ edges can be split into multiple copies such that all the copies have $0$ or $k$ edges, except for one copy which is allowed to have any number of edges in $\{0,1,\dots, k\}$. A solution of the original instance is then easily recovered from a solution of this modified instance. Note that since the set of neighbours is given as the output of a circuit, it will have length bounded by some polynomial in the input size and so this argument can indeed be applied.
\end{remark}

\begin{definition}[\ppak{k} \cite{papadimitriou1994complexity}]
For any $k \geq 2$, the class \ppak{k} is defined as the set of all \textup{TFNP} problems that many-one reduce to \bipa{k}.
\end{definition}

As a warm-up let us show the following:

\begin{proposition}[\cite{papadimitriou1994complexity}]
\ppak{2} $=$ \textup{PPA}
\end{proposition}

\begin{proof}
Recall that PPA can be defined using the canonical complete problem \textsc{Leaf}~\cite{papadimitriou1994complexity,beame1998relative}: given an undirected graph where every vertex has degree at most 2, and a leaf (i.e., degree $= 1$), find another leaf. This immediately yields \ppak{2} $\subseteq$ \textup{PPA}, since \bipa{2} is just a special case of \textsc{Leaf} where the graph is bipartite.

Given an instance of \textsc{Leaf} with graph $G=(\{0,1\}^n,E)$ we construct an instance of \bipa{2} on the vertex set $\{0,1\} \times \{0,1\}^{2n}$ as follows. For any $u \in \{0,1\}^n$ we have a vertex $x_u := 0u0^n$ on the left side of the bipartite graph. For any edge $\{u,v\} \in E$ ($u,v$ ordered lexicographically) we have a vertex $y_{uv}:=1uv$ on the right side of the bipartite graph and we create the edges $\{x_u,y_{uv}\}$ and $\{x_v,y_{uv}\}$. All other vertices in $\{0,1\} \times \{0,1\}^{2n}$ are isolated. In polynomial time we can construct a circuit that computes the neighbours of any vertex. Furthermore, $w \in \{0,1\}^n$ is a leaf, if and only if $x_w$ has degree 1. Finally, all vertices on the right-hand side have degree 0 or 2.
\end{proof}

\subsection{\texorpdfstring{$\bm{\ppakl{k}{\ell}}$}{PPA-k[\#l]} : Fixing the degree of the trivial solution}
In the definition of the \ppak{k}-complete problem \bipa{k} (\cref{def:bipartite-k}) the degree of the trivial solution $00^n$ can be any number in $\{1, \dots, k-1\}$. In this section we define more refined classes where the degree of the trivial solution is fixed. In \cref{sec:relationship}, these classes will be very useful to describe how the \ppak{k} classes relate to each other. These definitions are inspired by the corresponding ``counting principles'' studied in Beame et al.~\cite{beame1996counting} that were also defined in a refined form in order to describe how they relate to each other. We believe that these refined classes will also be useful to capture the complexity of natural problems. Note that for $k=2$, the degree of the trivial solution will always be $1$ and thus the question does not even appear in the study of PPA.

\begin{definition}
Let $k \geq 2$ and $1 \leq \ell \leq k-1$. The problem \bipal{k}{\ell} is defined as \bipa{k} (\cref{def:bipartite-k}) but with the additional condition $|C(00^n)| = \ell$.
\end{definition}
Note that this problem remains in TFNP, since the condition can be checked efficiently.

\begin{definition}[\ppakl{k}{\ell}]
Let $k \geq 2$ and $1 \leq \ell \leq k-1$. The class \ppakl{k}{\ell} is defined as the set of all \textup{TFNP} problems that many-one reduce to \bipal{k}{\ell}.
\end{definition}

If $k$ is some prime $p$, then these classes are not interesting. Indeed, it holds that \ppakl{p}{\ell} $=$ \ppak{p} for all $1 \leq \ell \leq p-1$. This can be shown using the following technique: take multiple copies of the instance and ``glue'' the trivial solutions together. If $p$ is prime, then any other degree of the glued trivial solution can be obtained (by taking the right number of copies). In fact this technique yields the stronger result:

\begin{lemma}\label{lem:gcdkell1-div-ell2}
If $\gcd(k,\ell_1)$ divides $\ell_2$, then \ppakl{k}{\ell_1} $\subseteq$ \ppakl{k}{\ell_2}.
\end{lemma}

\begin{proof}
Since $\gcd(k,\ell_1)$ divides $\ell_2$, there exists $m < k$ such that $m \times \ell_1 = \ell_2 \mod k$. Given an instance of \bipal{k}{\ell_1}, take the union of $m$ copies of the instance, i.e., $m2^n$ vertices on each side (and any additional isolated vertices needed to reach a power of $2$). Then, merge the $m$ different copies of the trivial solution into one (by redirecting edges to a single one). This vertex will have degree $m \ell_1 = \ell_2 \mod k$. Finally, apply the usual trick to ensure all degrees are in $\{0,1, \dots, k\}$ (\cref{rem:degin0k}).
\end{proof}

In particular, we also get the nice result \ppakl{k}{\ell} $=$ \ppakl{k}{\gcd(k,\ell)}. Applying the result to the case $k=6$, we get that \ppakl{6}{1} $=$ \ppakl{6}{5}, \ppakl{6}{2} $=$ \ppakl{6}{4}, as well as \ppakl{6}{1} $\subseteq$ \ppakl{6}{2} and \ppakl{6}{1} $\subseteq$ \ppakl{6}{2}. Thus, we have three ``equivalence classes'' $\{1,5\}$, $\{2,4\}$ and $\{3\}$ and the relationships $\{1,5\} \leq \{2,4\}$ and $\{1,5\} \leq \{3\}$. In \cref{sec:relationship}, we will show that $\{2,4\}$ corresponds to \ppak{3}, $\{3\}$ to \ppak{2} and $\{1,5\}$ to \ppak{2} $\cap$ \ppak{3}.

Now let us introduce some notation that will allow us to precisely describe the relationship between \ppak{k} and the \ppakl{k}{\ell}.
\begin{definition}[$\amper$ operation~\cite{buss2012propositional}]
Let $R_0$ and $R_1$ be two TFNP problems. Then the problem $R_0 \amper R_1$ is defined as: given an instance $I_0$ of $R_0$, an instance $I_1$ of $R_1$ and a bit $b \in \{0,1\}$, find a solution to $I_b$.
\end{definition}

This operation is commutative and associative (up to many-one equivalence). Indeed, $R_0 \amper R_1$ is many-one equivalent to $R_1 \amper R_0$, and $(R_0 \amper R_1) \amper R_2$ is many-one equivalent to $R_0 \amper (R_1 \amper R_2)$. Since the $\amper$ operation is associative, the problem $\amper_{\ell=1}^{k} R_\ell$ is well-defined up to many-one equivalence. It is also equivalent to the following problem: given instances $I_1, \dots, I_k$ of $R_1, \dots, R_k$ and an integer $j \in \{1,\dots,k\}$, find a solution to $I_j$.

We extend the $\amper$ operation to TFNP subclasses in the natural way. Let $C_0$ and $C_1$ be TFNP subclasses with complete problems $R_0$ and $R_1$ respectively. Then $C_0 \amper C_1$ is the class of all TFNP problems that many-one reduce to $R_0 \amper R_1$. Note that the choice of complete problems does not matter. Intuitively, this class contains all problems that can be solved in polynomial time by a Turing machine with a single oracle query to either $C_0$ or $C_1$.

Using this definition we obtain:

\begin{lemma}\label{lem:ppakl-amper}
For all $k \geq 2$ we have \ppak{k} $=$ $\amper\limits_{\ell=1}^{k-1}$ \ppakl{k}{\ell}.
\end{lemma}

\begin{proof}
One containment immediately follows from the fact that $\ppakl{k}{\ell} \subseteq \ppak{k}$ for all $\ell \in \{1, \dots, k-1\}$. For the other containment, note that for any instance $I$ of \bipa{k} we can easily compute $\ell \in \{1,\dots,k-1\}$ such that $I$ is also an instance of \bipal{k}{\ell}.
\end{proof}

Together with \cref{lem:gcdkell1-div-ell2}, \cref{lem:ppakl-amper} yields, e.g., \ppak{6} $=$ \ppakl{6}{2} $\amper$ \ppakl{6}{3}.

\section{Equivalent Definitions}\label{sec:complete-toolbox}

In this section we show that \ppak{k} can be defined by using other problems instead of \bipa{k}. The totality of these problems is again based on arguments modulo $k$. By showing that these problems are indeed \ppak{k}-complete, we provide additional support for the claim that \ppak{k} captures the complexity of ``polynomial arguments modulo $k$''. While these problems are not ``natural'' and thus not interesting in their own right, they provide equivalent ways of defining \ppak{k}, which can be very useful when working with these classes. In particular, we make extensive use of these equivalences in this work.

The TFNP problems we consider are the following:
\begin{itemize}
    \item \imba{k} : given a directed graph and a vertex that is \emph{unbalanced-mod-$k$}, i.e., out-degree $-$ in-degree $\neq 0 \mod k$, find another such vertex.
    \item \hyper{k} : given a hypergraph and a vertex that has degree $\neq 0 \mod k$, find another such vertex or a hyperedge that has size $\neq k$.
    \item \group{k} : given a set of size $\neq 0 \mod k$ and a partition into subsets, find a subset that has size $\neq k$.
\end{itemize}
As usual, the size of the graph (respectively hypergraph, set) can be exponential in the input size, and the edges (resp.\ hyperedges, subsets) can be computed efficiently locally. We also define the corresponding problems \imbal{k}{\ell}, \hyperl{k}{\ell} and \groupl{k}{\ell} analogously. The formal definitions of all these problems are provided in \cref{sec:complete-defs}.

\begin{theorem}\label{thm:complete-problems}
Let $k \geq 2$ and $1 \leq \ell \leq k-1$.
\begin{itemize}
    \item \imbal{k}{\ell}, \hyperl{k}{\ell}, \groupl{k}{\ell} are \ppakl{k}{\ell}-complete,
    \item \imba{k}, \hyper{k}, \group{k} are \ppak{k}-complete.
\end{itemize}
\end{theorem}

In his online post~\cite{jerabekstackexchange}, Je{\v r}{\'a}bek proves that \bipa{3}, \imba{3} and \group{3} are equivalent and (correctly) claims that the proof generalises to any other prime. Thus, our contribution is the definition of the problems for any $k \geq 2$ (and the $\ell$-parameter versions) and the generalisation of the result to any $k \geq 2$ (not only primes) and to the $\ell$-parameter versions of the problems, as well as to the new problem \hyper{k}.
The proof of \cref{thm:complete-problems} can be found in \cref{sec:complete-problems-proof}.

The problem \imba{k} is a generalisation of the PPAD-complete problem \textsc{Imbalance}~\cite{beame1998relative,GH2019hairy} : given a directed graph and a vertex that is unbalanced (i.e., out-degree $-$ in-degree $\neq 0$), find another unbalanced vertex. It is known~\cite{GH2019hairy} that in \textsc{Imbalance} we can assume without loss of generality that the given vertex has imbalance exactly $1$. As a result, \textsc{Imbalance} trivially reduces to \imba{k}, and thus \cref{thm:complete-problems} also yields\footnote{This observation was also made by Je{\v r}{\'a}bek for the classes \ppak{p} ($p$ prime).}:

\begin{corollary}
For all $k \geq 2$, we have \textup{PPAD} $\subseteq$ \ppak{k}.
\end{corollary}

Furthermore, if we use the convention that $a = b \mod 0$ if and only if $a=b$, then \imba{0} actually corresponds to \textsc{Imbalance}. Thus, in a certain sense we could define $\ppak{0} = \ppad$. On the other hand, \imba{1} is a trivial problem.

\subsection{Formal definitions}\label{sec:complete-defs}

\paragraph*{Imbalance.}
A directed graph on the vertex set $\{0,1\}^n$ is represented by Boolean circuits $S,P : \{0,1\}^n \to \set_{\leq k}(\{0,1\}^n)$ that output the successor and predecessor set of a given vertex, respectively. As before, it is enough to consider the case where the in- and out-degree of any vertex is at most $k$, since the general case reduces to this (analogously to \cref{rem:degin0k}). We syntactically enforce that $x \notin S(x) \cup P(x)$ and we interpret $S(x)$ (respectively, $P(x)$) as the set of potential successors (respectively, predecessors) of $x$. There is a directed edge from $x$ to $y$ if $y \in S(x)$ and $x \in P(y)$. The following problem was defined by Je{\v r}{\'a}bek~\cite{jerabekstackexchange}, but only for prime $k$ and without the $\ell$-parameter version.

\begin{definition}
Let $k \geq 2$. The problem \imba{k} is defined as: given Boolean circuits $S,P : \{0,1\}^n \to \set_{\leq k}(\{0,1\}^n)$ representing a directed graph on the vertex set $\{0,1\}^n$ with $|S(0^n)|-|P(0^n)| \notin \{-k,0,k\}$, find
\begin{itemize}
    \item $x \neq 0^n$ such that $|S(x)|-|P(x)| \notin \{-k,0,k\}$
    \item or $x,y$ such that $y \in S(x)$ but $x \notin P(y)$, or $y \in P(x)$ but $x \notin S(y)$.
\end{itemize}
For $1 \leq \ell \leq k-1$, \imbal{k}{\ell} is defined with the additional condition $|S(0^n)|-|P(0^n)| = \ell$.
\end{definition}

\paragraph*{Hypergraph.}
A hypergraph on the vertex set $\{0,1\}^n$ is represented as follows. For every vertex $x \in \{0,1\}^n$, a circuit $C: \{0,1\}^n \to \set_{\leq k}(\set_{\leq k}(\{0,1\}^n))$ outputs the set $C(x)$ of all hyperedges containing $x$, where each hyperedge is a set of vertices in $\{0,1\}^n$. As usual, we only need to consider the case where every vertex is contained in at most $k$ hyperedges and every hyperedge has size at most $k$. A hyperedge $\{x_1, \dots, x_m\}$ exists in the hypergraph, if all the vertices involved indeed agree that it is present, i.e., if $\{x_1, \dots, x_m\} \in C(x_i)$ for all $i \in \{1,\dots,m\}$.
\begin{definition}
Let $k \geq 2$. The problem \hyper{k} is defined as: given a Boolean circuit $C: \{0,1\}^n \to \set_{\leq k}(\set_{\leq k}(\{0,1\}^n))$ representing a hypergraph on the vertex set $\{0,1\}^n$ with $|C(0^n)| \notin \{0,k\}$, find
\begin{itemize}
    \item $x \neq 0^n$ such that $|C(x)| \notin \{0,k\}$
    \item or $x$ such that $C(x)$ contains a hyperedge of size $\neq k$
    \item or $x,y$ such that $C(x)$ and $C(y)$ are not consistent with one another.
\end{itemize}
For $1 \leq \ell \leq k-1$, \hyperl{k}{\ell} is defined with the additional condition $|C(0^n)| = \ell$.
\end{definition}
Note that for $k=2$ this problem essentially corresponds to the \ppa-complete problem \textsc{Leaf} and its (equivalent) generalisation \textsc{Odd}~\cite{papadimitriou1994complexity,beame1998relative}: given an undirected graph and a vertex with odd degree, find another one.

\paragraph*{Partition.}
A partition of $\{0,1\}^n$ is represented by a Boolean circuit $C: \{0,1\}^n \to \{0,1\}^n$ as follows: $x \in \{0,1\}^n$ lies in the subset given by the orbit of $x$ with respect to $C$, i.e., $\{C^i(x): i \geq 0\}$, where $C^i(x) = C(C(\dots C(x)) \dots)$ ($i$ times). The problem we define below is based on the simple observation that a base set of size $\neq 0 \mod k$ cannot be partitioned into sets of size $k$. The base set consists of all elements in $\{0,1\}^n$ except for $m$ elements that have been removed, for some $m < 2^n$ such that $2^n-m \neq 0 \mod k$. Here it is convenient to identify $\{0,1\}^n$ with $\{0,1, \dots, 2^n-1\}$ in the natural way. Thus, we can think of the base set as simply being $\{m,m+1, \dots, 2^n-1\}$.

\begin{definition}[\group{k}]
Let $k \geq 2$. The problem \group{k} is defined as: given $m < 2^n$ with $2^n-m \neq 0 \mod k$ and a Boolean circuit $C: \{0,1\}^n \to \{0,1\}^n$, such that $C(x)=x$ for all $x < m$, find
\begin{itemize}
    \item $x \geq m$ and $d \in \mathbb{N}$ such that $C^d(x)=x$ and $d | k$, $d \neq k$
    \item or $x \in \{0,1\}^n$ such that $C^k(x) \neq x$
\end{itemize}
where $d | k$ means that $d$ divides $k$.

For $1 \leq \ell \leq k-1$, \groupl{k}{\ell} is defined with the additional condition $2^n-m = \ell \mod k$.
\end{definition}

The condition ``$C(x)=x$ for all $x < m$'' corresponds to excluding elements that do not lie in the base set and it can be enforced syntactically. The first solution type corresponds to finding a set in the partition such that its size divides $k$ (but is $\neq k$), while the second solution type corresponds to finding a set such that its size does not divide $k$. Note that a solution is guaranteed to exist since $2^n-m \neq 0 \mod k$.

The definition of this problem is inspired by the \modd{k} problems defined by Buss and Johnson~\cite{buss2012propositional} (for prime $k>2$) and by Johnson~\cite{johnson2011thesisNPsearch} (for any $k \geq 2$). In \cref{sec:pmod} we argue that, unlike the problem defined above, the \modd{k} problems only partially capture the complexity of arguments modulo $k$ (when $k$ is not a prime power). The problem was also defined by Je{\v r}{\'a}bek~\cite{jerabekstackexchange}, but only for $k$ prime and without the $\ell$-parameter version. Finally, note that \group{2} essentially corresponds to the \ppa-complete problem \textsc{Lonely}~\cite{beame1998relative}.

The definition of this problem can be modified in various ways without changing its complexity. For instance, the first solution type can be changed to simply ask for $x \geq m$ such that $C^d(x) = x$ for some $d < k$. We have defined the problem in a slightly more complicated way to make the connection with the \modd{k} problems more immediate (see \cref{sec:pmod}). Yet another equivalent way of defining the problem would be to consider a Boolean circuit $C: \{0,1\}^n \to \set_{\leq k}(\{0,1\}^n)$ where $C(x) \subseteq \{0,1\}^n$ is interpreted as the set containing $x$ in the partition. A solution would then be any $x \geq m$ with $|C(x)| < k$ or any $x,y$ witnessing an inconsistency in the partition given by $C$.

\subsection{Proof of \texorpdfstring{\cref{thm:complete-problems}}{Theorem~\ref*{thm:complete-problems}}}\label{sec:complete-problems-proof}

We omit some details that are easy to fill in. For example, when given an instance of \imbal{k}{\ell}, we assume that all the edges are well-defined, i.e., solutions of the second type never occur. Indeed, given a generic instance of the problem, it can be reduced to an instance where this holds by modifying the circuits so that they check and correct the successor/predecessor list before outputting it. Note that in the new instance only solutions of the first type can occur, but they can yield a solution of the second type of the original problem. The same observation also holds for \bipal{k}{\ell} (edges well-defined), \hyperl{k}{\ell} (hyperedges well-defined) and \groupl{k}{\ell} (size of any subset divides $k$).

\paragraph*{$\bm{\bipal{k}{\ell} \leq \hyperl{k}{\ell}}$ :}
We construct a hypergraph on the vertex set $\{0,1\}^n$. We identify every vertex $u$ of the hypergraph with the vertex $0u$ on the left-hand side of the bipartite graph. The hyperedges are given by the vertices on the right-hand side of the bipartite graph. More precisely, if for every right-hand side vertex $1v$ we let $N(1v)$ be the set of neighbours (on the left-hand side), then the set of hyperedges is exactly $\{N(1v): v \in \{0,1\}^n\}$. Note that given $u \in \{0,1\}^n$, we can find all the hyperedges containing $u$ in polynomial time. Furthermore, since the vertex $00^n$ has degree $\ell$ in the bipartite graph, the corresponding vertex $0^n$ in the hypergraph will also have degree $\ell$. It is easy to check that any solution of the \hyperl{k}{\ell} instance (in particular also any hyperedge that does not have size $k$) yields a solution to the \bipal{k}{\ell} instance.

\paragraph*{$\bm{\hyperl{k}{\ell} \leq \imbal{k}{\ell}}$ :} We construct a directed graph on the vertex set $\{0,1\}^{kn+2}$. For each vertex $u$ of the hypergraph there is a vertex $v_u$ in the directed graph (e.g., $v_u = 0u0^{(k-1)n+1}$) and for each hyperedge $e$ of the hypergraph there is a vertex $v_e$ in the directed graph (e.g., $v_e = 1u_1 \dots u_m10^{(k-m)n}$ where $e= \{u_1, \dots, u_m\}$ is ordered lexicographically). We put a directed edge from $v_u$ to $v_e$ iff $u \in e$ (i.e., iff $e$ appears in the hyperedge list of $u$). All other vertices are isolated. Note that $0^{kn+2} = v_{0^n}$ has imbalance $\ell$ and for any vertex in $\{0,1\}^{kn+2}$ we can compute the predecessors and successors in polynomial time. If there is an imbalance modulo $k$ in a vertex $v_e$, then the corresponding hyperedge $e$ does not have size $k$. If there is an imbalance modulo $k$ in a vertex $v_u$, then $u$ has degree $\neq 0 \mod k$ in the hypergraph.

\paragraph*{$\bm{\imbal{k}{\ell} \leq \groupl{k}{\ell}}$ :}
Consider an instance of the problem \imbal{k}{\ell}. Split every vertex $v$ into two vertices $v_{in}$ and $v_{out}$, such that $v_{in}$ gets all the incoming edges and $v_{out}$ gets all the outgoing edges. If $v$ was balanced, i.e., in-deg($v$) $=$ out-deg($v$) $= d$, then we add $k-d$ edges from $v_{out}$ to $v_{in}$. We can assume that out-deg($0^n$) $= \ell$ and in-deg($0^n$) $= 0$ (just create a copy of $0^n$ that takes in-deg($0^n$) incoming and outgoing edges), and thus out-deg($0^n_{out}$) $= \ell$ and in-deg($0^n_{in}$) $= 0$. Note that we are using multi-edges, which are not allowed in the definition of the problem. However, this is not an issue, since this is just an intermediate step of the reduction. This new instance has the property that no vertex has both incoming and outgoing edges. Furthermore, any solution (i.e., a vertex with in- or out-degree not in $\{0,k\}$, except $0^n_{out}$) yields a solution of the original instance.

Thus, we can assume wlog that the \imbal{k}{\ell} instance (with multi-edges) is such that no vertex has both incoming and outgoing edges. We construct an instance of \groupl{k}{\ell} on the set $\{0,1\}^{k + n}$. Every vertex $u$ of the directed graph has $k$ corresponding elements in the set $\{0,1\}^{k + n}$, namely $u_1 = 10^{k-1}u$, $\dots$, $u_k = 0^{k-1}1u$. If $u$ does not have any outgoing edges, then $u_1, \dots, u_k$ form a subset of the partition, i.e., $C(u_i) = u_{i+1 \mod k}$. If $u$ has outgoing edges to $v^{(1)}, v^{(2)}, \dots, v^{(j)}$ ($j \leq k$, ordered lexicographically), then for every $i=1, \dots, j$ we put $u_i$ in a subset that we denote $S_{v^{(i)}}$. $u_{j+1}, \dots, u_k$ are put into isolated subsets, i.e., $C(u_i)=u_i$ for all $i = j+1, \dots, m$. Note that if $v$ has $k$ incoming edges, then $S_v$ will contain $k$ elements. Given any element $u_i$, we can compute all the elements in its subset in polynomial time (and thus efficiently construct $C$ that cycles through them in lexicographic order). Furthermore, since out-deg($0^n$) $= \ell$, the vertices $0^n_{\ell+1}, \dots, 0^n_k$ will be in singleton sets. Consider the subset of $\{0,1\}^{k+n}$ $X = \{u_i : u \in \{0,1\}^n, i \in \{1, \dots, k\}\} \setminus \{0^n_{\ell+1}, \dots, 0^n_k\}$. Then $|X| = k 2^n - (k-\ell) = \ell \mod k$. It is easy to check that any element in $X$ that is not contained in a subset of size $k$ (according to $C$), must yield a solution to the \imbal{k}{\ell} instance. Finally, the last step is to construct an efficient bijection between $X$ and the set of all integers $\{j : 2^{n+k} - |X| \leq j < 2^{n+k}\}$, which is easy to do. Thus, we have reduced the original instance to an instance of \groupl{k}{\ell} with inputs $m=2^{n+k}-|X|$ and $C$ (modified according to the bijection).

\paragraph*{$\bm{\groupl{k}{\ell} \leq \bipal{k}{\ell}}$ :}
Let us consider any instance $(C,m)$ of \groupl{k}{\ell} with parameter $n$. In particular, it holds that $m < 2^n$ and $2^n-m = \ell \mod k$. We construct a bipartite graph as follows. The vertex sets on the left and right side are $A = \{0,1\}^n$ and $B = \{0,1\}^n$ respectively. We can define a canonical partition of the numbers $m, m+1, \dots, 2^n-1$ into sets of size $k$ (and one set of size $\ell$). For example, $\{m, m+1, \dots, m+ \ell -1\}$, $\{m+ \ell, m+\ell + 1, \dots, m+\ell + k-1\}$, etc. Each set of the canonical partition corresponds to a vertex in $A$ as follows: a set containing $k$ numbers $x_1 < x_2 < \dots < x_k$ is represented by the vertex $x_1 \in A$. For the set of size $\ell$ in the canonical partition, we introduce a special case: it is represented by $0^n \in A$. Note that many vertices in $A$ will not correspond to any set of the canonical partition. For a number $x \geq m$, we let $L(x) \in A$ denote the vertex in $A$ representing the set containing $x$ in the canonical partition.

Another partition of the numbers $m, m+1, \dots, 2^n-1$ into sets of size at most $k$ is given by the instance $(C,m)$ of \groupl{k}{\ell}. Similarly to what we did above, we can associate each set in the partition given by $C$ with a vertex in $B$. We let $R(x) \in B$ denote the vertex of $B$ representing the set containing $x$ in the partition given by $C$. Note that for any vertex in $A$ or $B$ we can efficiently determine whether it represents a set of one of the two partitions and if so, which set exactly it represents.

For every $x \geq m$ we add an edge between $L(x) \in A$ and $R(x) \in B$, i.e., between the sets that contain $x$ in the two different partitions. This construction might introduce multi-edges (if some $x$ and $y$ lie in the same set in both partitions) but this can easily be resolved by using the \emph{Mitosis gadgets} described below. It is easy to see that for any vertex of the bipartite graph we can efficiently compute the set of all its neighbours. Finally, note that the vertex $0^n \in A$ has degree $\ell$, and any other vertex with degree $\neq 0 \mod k$ must necessarily lie in $B$ and correspond to a set in the partition given by $C$ that contains strictly less than $k$ elements. Thus any such vertex immediately yields a solution to the original \groupl{k}{\ell} instance.

\paragraph*{Mitosis gadgets.}
Let $k \geq 2$. We now show how to construct a small bipartite graph such that exactly one vertex on each side has degree $1$ and all other vertices have degree $k$ (or $0$). This ``gadget'' can then be used to increase the degree of two vertices (one on each side of the bipartite graph) without adding any solutions, i.e., vertices with degree $\neq 0 \mod k$.

The gadget is a bipartite graph with $k+1$ vertices on each side: $a_1, \dots, a_{k+1}$ and $b_1, \dots, b_{k+1}$. It contains all the edges $\{a_i,b_j\}$ for $i,j \leq k$, except the edge $\{a_k,b_k\}$. It also contains the edges $\{a_k,b_{k+1}\}$ and $\{a_{k+1},b_k\}$. Thus, all vertices have degree $k$, except for $a_{k+1}$ and $b_{k+1}$ which have degree $1$.

We call this the ``Mitosis'' gadget, because it allows us to duplicate edges that already exist. Let $u$ and $v$ be two vertices in a bipartite graph, one on each side. Furthermore, consider the case where there is an edge $\{u,v\}$. We would like to increase the degree of $u$ and $v$ by $1$, but without introducing any new solutions, in particular without introducing any vertex with degree $\neq 0 \mod k$. Using the Mitosis gadget, we can just add new vertices $a_1, \dots, a_k$ and $b_1, \dots, b_k$, and identify $a_{k+1}$ with $u$ and $b_{k+1}$ with $v$. Adding the corresponding vertices of the gadget yields a bipartite graph where the degree of $u$ and $v$ has increased by $1$, but no new solutions have been introduced. Note that this gadget can, in particular, be used to turn a bipartite graph with multi-edges into one without them, without changing the degree of existing vertices and without adding any new solutions.

\section{Relationship Between the Classes}\label{sec:relationship}

In this section, we present some results that provide deeper insights into how the classes relate to each other. For any $k \geq 2$, $\pfactors{k}$ denotes the set of all prime factors of $k$. The main conceptual result is that \ppak{k} is entirely determined by the set of prime factors of $k$:

\begin{theorem}\label{thm:ppak-structure}
For any $k \geq 2$ we have \ppak{k} $=$ $\amper\limits_{p \in \pfactors{k}}$ \ppak{p}.
\end{theorem}

\noindent This equation can be understood as saying the following:
\begin{itemize}
    \item Given a single query to an oracle for \ppak{k}, we can solve any problem in \ppak{p} for any $p \in \pfactors{k}$
    \item Given a single query to an oracle that solves any \ppak{p} problem for any $p \in \pfactors{k}$, we can solve any problem in \ppak{k}.
\end{itemize}

\begin{corollary}\label{cor:ppak-structure}
In particular, we have:
\begin{itemize}
    \item For $k_1,k_2 \geq 2$, if $\pfactors{k_1} \subseteq \pfactors{k_2}$, then \ppak{k_1} $\subseteq$ \ppak{k_2}.
    \item For all $k_1,k_2 \geq 2$, \ppak{k_1k_2} $=$ \ppak{k_1} $\amper$ \ppak{k_2}.
    \item For all $k \geq 2$ and all $r \geq 1$ we have \ppak{k^r} $=$ \ppak{k}.
\end{itemize}
\end{corollary}

Using the \ppakl{k}{\ell} classes, we can formulate an even stronger and more detailed result. For any $k \geq 2$, $1 \leq \ell \leq k-1$, we define $\pfactors{k,\ell} = \pfactors{k/\gcd(k,\ell)}$. In this case the conceptual result says that \ppakl{k}{\ell} is entirely determined by the set of prime factors of $k/\gcd(k,\ell)$.

\begin{theorem}\label{thm:ppakl-structure}
For any integer constants $k$ and $\ell$ with $k \geq 2$ and $0 < \ell < k$, it holds that
$$\ppakl{k}{\ell} = \ppakl{\left(\prod_{p \in \pfactors{k,\ell}} p\right)}{1} = \bigcap_{p \in \pfactors{k,\ell}} \ppak{p}.$$
\end{theorem}
The proof of \cref{thm:ppakl-structure} can be found in the next section. Before we move on to that, let us briefly show that \cref{thm:ppak-structure} follows from \cref{thm:ppakl-structure}.

\begin{proof}[Proof of \cref{thm:ppak-structure}]
Using \cref{lem:ppakl-amper} and \cref{thm:ppakl-structure} we can write
$$\ppak{k} = \amper\limits_{\ell=1}^{k-1} \ppakl{k}{\ell} = \amper\limits_{\ell=1}^{k-1} \left( \bigcap_{p \in \pfactors{k,\ell}} \ppak{p} \right) = \amper\limits_{p \in \pfactors{k}} \ppak{p}$$
where the last equality follows by noting that $\pfactors{k,\ell} \subseteq \pfactors{k}$ for all $\ell$, and $\pfactors{k,k/p} = \{p\}$ for all $p \in \pfactors{k}$.
\end{proof}

\subsection{Proof overview}

\begin{proof}[Proof of \cref{thm:ppakl-structure}]
All containment results follow from \cref{thm:ppakl-char} below, except
$$\ppakl{\left(\prod_{p \in \pfactors{k,\ell}} p\right)}{1} \supseteq \bigcap_{p \in \pfactors{k,\ell}} \ppak{p}.$$

Let $\pfactors{k,\ell} = \{p_1, \dots, p_d\}$. We will show how to combine a set of instances $(C_1,m_1),$ $\dots,$ $(C_d,m_d)$, where $(C_i,m_i)$ is an instance of \groupl{p_i}{1}, into a single instance of \groupl{s}{1}, where $s=p_1p_2 \cdots p_d$, such that any solution to this instance yields a solution to one of the $(C_i,m_i)$ instances. Without loss of generality, we can assume that the parameter $n$ is the same for all $(C_i,m_i)$ instances. Without loss of generality, we can assume that $2^n-m_i = 1 \mod s$ for all $i$, because we can add at most $\prod_{j \neq i} p_i$ sets of size $p_i$ to achieve this (see the proof of \cref{lem:ppa-kll_in_ppa-k1}). Note that we then have $(2^n-m_1)(2^n-m_2) \cdots (2^n-m_d) = 1 \mod s$. Furthermore, for $(x_1, \dots, x_d)$ and $(y_1, \dots, y_d)$ with $x_i,y_i \geq m_i$ we can define $x \equiv y$ if and only if $x_i \equiv_i y_i$ for all $i$, where $x_i \equiv_i y_i$ means that $x_i$ and $y_i$ lie in the same set in instance $(C_i,m_i)$. If for all $i$, $x_i$ lies in a set of size $p_i$ in $(C_i,m_i)$, then $x$ will lie in a set of size $s$. Thus any solution yields a solution to one of the original instances. The details to fully formalise this are very similar to the proof of \cref{lem:ppa-k1_in_ppa-r1}.
\end{proof}

In \cite{beame1996counting} Beame et al.\ investigated the relative proof complexity of so-called ``counting principles''. These counting principles are formulas that represent the fact that a set of size $\neq 0 \mod k$ cannot be partitioned into sets of size $k$. They investigated the relationship between these principles in terms of whether one can be proved from the other by using a constant-depth, polynomial-size Frege proof. Their main result is a full characterisation of when this is possible or not. As noted by Johnson~\cite{johnson2011thesisNPsearch}, these counting formulas do not yield NP search problems, but they can be related to corresponding NP search problems (TFNP, in fact). Indeed, Johnson uses this connection to obtain some separation results between his \ppmod{k} classes (see \cref{sec:pmod}) from Beame et al.'s negative results. Our contribution is using Beame et al.'s positive results in order to prove inclusion results about the \ppakl{k}{\ell} classes. More precisely, we modify their proofs to obtain polynomial-time reductions between our \groupl{k}{\ell} problems.
Thus, we obtain the following analogous result:

\begin{theorem}\label{thm:ppakl-char}
Let $k_1,k_2 \geq 2$ and $0 < \ell_i < k_i$ for $i=1,2$. If $\pfactors{k_2,\ell_2} \subseteq \pfactors{k_1,\ell_1}$, then $\ppakl{k_1}{\ell_1} \subseteq \ppakl{k_2}{\ell_2}$.
\end{theorem}

\begin{proof}
From \cref{lem:gcdkell1-div-ell2} we know that \ppakl{k_i}{\ell_i} = \ppakl{k_i}{\gcd(k_i,\ell_i)} for $i=1,2$. The result then follows from a few technical lemmas proved in \cref{app:sec:char-lemmas}:

\begin{equation*}
\begin{split}
\ppakl{k_1}{\gcd(k_1,\ell_1)} \underset{\text{L~\ref{lem:ppa-kll_in_ppa-k1}}}{\subseteq} \ppakl{\frac{k_1}{\gcd(k_1,\ell_1)}}{1} &\underset{\text{L~\ref{lem:ppa-k1_in_ppa-r1}}}{\subseteq} \ppakl{\frac{k_2}{\gcd(k_2,\ell_2)}}{1}\\
&\underset{\text{L~\ref{lem:ppa-kl_in_ppa-krlr}}}{\subseteq} \ppakl{k_2}{\gcd(k_2,\ell_2)}
\end{split}
\end{equation*}
\end{proof}

\section{Johnson's \texorpdfstring{$\bm{\ppmod{k}}$}{PMODk} Classes and Oracle Separations}\label{sec:pmod}

Inspired by the definition of the PPA-complete problem \textsc{Lonely}~\cite{beame1998relative}, Buss and Johnson~\cite{buss2012propositional} defined TFNP problems called \modd{p} to represent arguments modulo some prime $p$. Their main motivation was to use these problems to show separations (in the type 2 setting) between Turing reductions with $m$ oracle queries and Turing reductions with $m+1$ oracle queries. In his thesis~\cite{johnson2011thesisNPsearch}, Johnson generalised the definition of \modd{k} to any $k \geq 2$ and defined corresponding classes \ppmod{k}. He also proved some separations between these classes and other TFNP classes in the type 2 setting (which yield oracle separations in the standard setting). It seems that Johnson was not aware of Papadimitriou's~\cite{papadimitriou1994complexity} \ppak{p} classes.

In this section, we study the classes \ppmod{k} and prove a characterisation in terms of the classes \ppak{p}. In particular, we show that \ppmod{k} does not capture the full strength of arguments modulo $k$, when $k$ is not a prime power. This characterisation also allows us to use Johnson's separations to obtain some oracle separations involving \ppak{k} and other TFNP classes.

Informally, the problem \modd{k} can be defined as follows. We are given a partition of $\{0,1\}^n$ into subsets and the goal is to find one of these subsets that has size $\neq k$. If $k$ is not a power of $2$, then such a subset must exist. If $k$ is a power of $2$, then we instead consider $\{0,1\}^n \setminus \{0^n\}$ and the problem remains total.

\begin{definition}[\modd{k}~\cite{buss2012propositional,johnson2011thesisNPsearch}]
Let $k \geq 2$. The problem \modd{k} is defined as: given a Boolean circuit $C$ with $n$ inputs and outputs,
\begin{itemize}
\item If $k$ is not a power of $2$: Find
\begin{itemize}
    \item $x \in \{0,1\}^n$ and $d \in \mathbb{N}$ such that $C^d(x)=x$ and $d | k$, $d \neq k$
    \item or $x \in \{0,1\}^n$ such that $C^k(x) \neq x$
\end{itemize}

\item If $k$ is a power of $2$: Let additionally $C(0^n)=0^n$ and find
\begin{itemize}
    \item $x \in \{0,1\}^n \setminus \{0^n\}$ and $d \in \mathbb{N}$ such that $C^d(x)=x$ and $d | k$, $d \neq k$
    \item or $x \in \{0,1\}^n$ such that $C^k(x) \neq x$
\end{itemize}
\end{itemize}
where $C^\ell(x) = C(C(\dots C(x)) \dots)$ ($\ell$ times) and $d | k$ means that $d$ divides $k$.
\end{definition}

\begin{definition}[\ppmod{k}~\cite{johnson2011thesisNPsearch}]
For any $k \geq 2$, the class \ppmod{k} is defined as the set of all \textup{TFNP} problems that many-one reduce to \modd{k}.
\end{definition}

Note that the problem \modd{k} is a special case of our problem \group{k} (which was indeed inspired by this definition). As a result, we immediately get that $\ppmod{k} \subseteq \ppak{k}$. Unless $k$ is a prime power, we don't expect this to hold with equality. The intuition is that restricting the size of the base set to always be a power $2$ has the effect of only achieving a subset of the possible $\ell$-parameter values of \ppakl{k}{\ell}. Namely, only $\ell \in \{2^n \mod k: n \in \mathbb{N}\}$ are achieved (for $k$ not a power of 2).

Johnson proves a lemma~\cite[Lemma 7.4.5]{johnson2011thesisNPsearch} that gives some idea of how the \ppmod{k} classes relate to each other. It can be stated as follows: if $k=p_1p_2 \dots p_r$, where the $p_i$ are distinct primes, then $\ppmod{k} = \cap_i \ppmod{p_i}$. He proves this if all $p_i \neq 2$ and claims that the proof also works if some $p_i=2$. However, if some $p_i=2$ then the proof does not work. This is easy to see, since our results below prove that $\ppmod{6} = \ppmod{3}$ which is not equal to $\ppmod{2} \cap \ppmod{3}$, unless $\ppmod{2} \subseteq \ppmod{3}$. However, Johnson proves that $\ppmod{2} \not \subseteq \ppmod{3}$ in the type 2 setting.

The following result provides a full characterisation of \ppmod{k} in terms of the classes \ppak{p}.

\begin{theorem}\label{thm:pmod-char}
Let $k \geq 2$.
\begin{itemize}
\item If $k$ is not a power of $2$, then $\ppmod{k} = \ppakl{\widetilde{k}}{1} = \cap_{p \in \pfactors{\widetilde{k}}} \ppak{p}$ where $\widetilde{k}$ is the largest odd divisor of $k$.
\item If $k$ is a power of $2$, then $\ppmod{k} = \ppak{2}$.
\end{itemize}
\end{theorem}

The proof of \cref{thm:pmod-char} is given below in \cref{sec:pmod-char-proof}.

\begin{corollary}
In particular, we have:
\begin{itemize}
\item for all primes $p$ and $r \geq 1$, $\ppmod{p^r} = \ppak{p^r} = \ppak{p}$
\item for all $k \geq 2$, $\ppmod{2k} = \ppmod{k}$
\item for all odd $k \geq 3$, $\ppmod{k} = \ppakl{k}{1} = \cap_{p \in \pfactors{k}} \ppak{p}$
\end{itemize}
\end{corollary}

If $k$ is a prime power, then \ppmod{k} is the same as \ppak{k}. However, for other values of $k$, we argue that \ppmod{k} fails to capture the full strength of arguments modulo $k$. For example, $\ppmod{15} = \ppakl{15}{1} = \ppak{3} \cap \ppak{5}$, whereas $\ppak{15} = \ppak{3} \amper \ppak{5}$. This means that \ppak{15} can solve any problem that lies in \ppak{3} or \ppak{5}, while \ppmod{15} can only solve problems that lie \emph{both} in \ppak{3} and \ppak{5}. In particular, if $\ppmod{15} = \ppak{15}$, then it would follow that $\ppak{3} = \ppak{5}$, which is not believed to hold (see oracle separations below). Even worse perhaps, is the fact that $\ppmod{2k} = \ppmod{k}$ for any $k \geq 2$. In particular, this means that $\ppmod{6} = \ppmod{3}$, which indicates that \ppmod{6} does not really capture arguments modulo $6$.

Nevertheless, Johnson's oracle separation results (obtained from the corresponding type 2 separations as in~\cite{beame1998relative}) also yield corresponding results for the \ppak{k} classes (using \cref{thm:pmod-char}). We briefly mention a few of the results obtained this way. See Johnson~\cite[Chapter 8]{johnson2011thesisNPsearch} for additional results. Relative to any generic oracle (see \cite{beame1998relative}):
\begin{itemize}
    \item $\ppak{p} \not \subseteq \ppak{q}$ for any distinct primes $p,q$
    \item $\ppak{k} \not \subseteq \textup{PPP}$, $\ppak{k} \not \subseteq \textup{PLS}$, $\ppak{k} \not \subseteq \textup{PPADS}$ for any $k \geq 2$
    \item $\textup{PPP} \not \subseteq \ppak{p}$, $\textup{PLS} \not \subseteq \ppak{p}$ for any prime $p$
\end{itemize}

\subsection{Proof of \texorpdfstring{\cref{thm:pmod-char}}{Theorem~\ref*{thm:pmod-char}}}\label{sec:pmod-char-proof}

For $k=2$, \modd{2} corresponds to the PPA-complete problem \textsc{Lonely}~\cite{beame1998relative}, and thus $\ppmod{2} = \ppa = \ppak{2}$. Let $r \geq 2$. Consider an instance $(C,m)$ of \groupl{2^r}{(2^r-1)} on the set $\{0,1\}^n$. Without loss of generality, assume $n \geq r$. Then $2^n = 0 \mod 2^r$ and thus $m = 2^n - (2^n-m) = -(2^r-1) \mod 2^r = 1 \mod 2^r$. This means that we can (efficiently) partition $\{0,1,\dots, m-1\}$ into subsets of size $2^r$, leaving only $0=0^n$ out. Thus, we have reduced \groupl{2^r}{(2^r-1)} to \modd{2^r}. Since $\ppakl{2^r}{(2^r-1)} = \ppak{2}$ (\cref{thm:ppakl-structure}), we obtain $\ppak2 \subseteq \ppmod{2^r}$. On the other hand we also have $\ppmod{2^r} \subseteq \ppak{2^r} = \ppak{2}$ by \cref{cor:ppak-structure}.

Consider some $k \geq 3$ that is not a power of $2$. First, let us show that $\ppmod{2k} = \ppmod{k}$. \modd{2k} reduces to \modd{k} by splitting every subset into two subsets of size $k$ (or less, if the subset has size $< 2k$). Conversely, consider an instance of \modd{k} on the set $\{0,1\}^n$. Make a copy of the instance, thus obtaining an instance on the set $\{0,1\}^{n+1}$. For every subset of the original instance, take the union with its copy. If the subset had size $k$, the new subset has size $2k$. Thus, we have reduced to \modd{2k}.

Let $k \geq 3$ be coprime with $2$. We will show $\ppmod{k} = \ppakl{k}{1}$. Consider an instance of \modd{k} on the set $\{0,1\}^n$. Since $k$ and $2$ are coprime, there exists $i \in \{0, \dots, k-1\}$ such that $2^{n+i} = 1 \mod k$ (e.g., by using Euler's theorem). Thus, we take $2^i$ copies of the instance and obtain an instance on the set $\{0,1\}^{n+i}$, which is an instance of \groupl{k}{1} (with $m=0$), since $2^{n+i} = 1 \mod k$. Conversely, consider an instance $(C,m)$ of \groupl{k}{1} on the set $\{0,1\}^n$. As before, there exists $i \in \{0, \dots, k-1\}$ such that $2^{n+i} = 1 \mod k$. We construct an instance $C'$ of \modd{k} on $\{0,1\}^{n+i}$ as follows. The element $x \in \{0,1\}^n$ of the original instance corresponds to the element $1^ix \in \{0,1\}^n$ of the new instance. If $x \geq m$, set $C'(1^ix) = 1^iC(x)$. The number of elements that have not yet been assigned to a subset is $m + (2^i-1)2^n = (m-2^n)+2^{n+i} = 0 \mod k$. Thus, we can efficiently partition them into subsets of size $k$ without introducing any solution. We have obtained an instance of \modd{k}.

\section{Many-one vs Turing Reductions}\label{sec:turing-reduction}

\begin{theorem}\label{thm:ppa-p-turingclosed}
For any prime $p \geq 2$, \ppak{p} is closed under Turing reductions.
\end{theorem}

In particular, \ppak{p^r} = \ppak{p} is also closed under Turing reductions. The proof of \cref{thm:ppa-p-turingclosed} can be found in \cref{sec:turing-closed-proof}. Furthermore, we also obtain:

\begin{corollary}\label{cor:ppa-kl-turingclosed}
For all $k \geq 2$ and $0 < \ell < k$, \ppakl{k}{\ell} is closed under Turing reductions.
\end{corollary}

\begin{proof}[Proof of \cref{cor:ppa-kl-turingclosed}]
Using \cref{thm:ppakl-structure}, we have $\ppakl{k}{\ell} = \bigcap_{p =1}^d \ppak{p_i}$, where we let $\{p_1, \dots, p_d\} = \pfactors{k,\ell}$. Consider a Turing reduction from some problem to \ppakl{k}{\ell}. Since $\ppakl{k}{\ell} \subseteq \ppak{p_i}$, this yields a Turing reduction to \ppak{p_i}, in particular. By \cref{thm:ppa-p-turingclosed}, it follows that there exists a many-one reduction to \ppak{p_i}, i.e., the problem lies in \ppak{p_i}. Since this holds for all $p_i$, the result follows.
\end{proof}

If $k$ is not a prime power, then it is not known whether \ppak{k} is closed under Turing reductions. Using our results from \cref{sec:pmod}, we can actually provide an oracle separation between \ppak{k} and the Turing-closure of \ppak{k}, i.e., an oracle under which \ppak{k} is not closed under Turing reductions. Let $R_1, \dots, R_k$ be TFNP problems. Following Johnson~\cite{johnson2011thesisNPsearch} we define $\bigotimes_{j=1}^k R_j$ as the problem: given instances $(I_1,\dots,I_k)$, where $I_j$ is an instance of $R_j$, solve $I_j$ for all $j$. As we did with the $\amper$ operation, with a slight abuse of notation, we can also use the operation $\otimes$ with the \ppak{k} classes. In \cite[Theorem 7.6.1]{johnson2011thesisNPsearch}, Johnson proved that for $m \geq 2$ and distinct primes $p_1, \dots, p_m$, $\bigotimes_{i=1}^m$ \modd{p_i} does not many-one reduce to $\amper_{i=1}^m$ \modd{p_i} in the type 2 setting. Together with our \cref{thm:ppak-structure,thm:pmod-char} this yields:

\begin{theorem}\label{thm:turing-separation}
Let $k \geq 2$ not a power of a prime. Relative to any generic oracle, it holds that $\bigotimes_{p \in \pfactors{k}} \ppak{p} \not\subseteq \ppak{k}$. In particular, relative to any generic oracle, \ppak{k} is not closed under Turing reductions.
\end{theorem}

$S = \bigotimes_{p \in \pfactors{k}} \ppak{p}$ corresponds to solving \ppak{p} for all prime factors $p$ of $k$ \emph{simultaneously}. In particular, this can be done by using $|\pfactors{k}|$ queries to \ppak{k}, i.e., a Turing reduction to \ppak{k}. Thus, $S$ lies in the Turing closure of \ppak{k}, but not in \ppak{k} (relative to any generic oracle).

\subsection{Proof of \texorpdfstring{\cref{thm:ppa-p-turingclosed}}{Theorem~\ref*{thm:ppa-p-turingclosed}}}\label{sec:turing-closed-proof}

We essentially apply the same technique that was used by Buss and Johnson~\cite{buss2012propositional} to show that PPA, PPAD, PPADS and PLS are closed under Turing reductions.

Let $\Pi$ be a problem that Turing-reduces to some problem in \ppak{p}. This means that there exists a Turing machine $M$ with access to a \ppak{p}-oracle that solves $\Pi$ in polynomial time. Since \imba{p} is \ppak{p}-complete (\cref{thm:complete-problems}), we assume that the oracle provides solutions to \imba{p} instances. Our goal is to show that all the oracle queries can be combined into a single one. Indeed, a Turing reduction that always uses a single oracle query immediately yields a many-one reduction. Thus, by the definition of \ppak{p}, this would yield $\Pi \in \ppak{p}$.

We begin by showing that any \imba{p}-instance can be efficiently transformed into an instance that has a particular form, namely: the starting node has imbalance $+1$ (in-degree $0$ and out-degree $1$), and any solution has imbalance $-1$ (in-degree $1$ and out-degree $0$). This can be achieved by the following steps:
\begin{enumerate}
    \item Ensure that all vertices have in- and out-degree at most $p$ (by splitting vertices into multiple copies).
    \item Ensure that any unbalanced vertex has in- or out-degree $0$ (by creating a copy that will take all the edges that yield the imbalance).
    \item Since $p$ is prime, we can ensure that the starting vertex has imbalance $+1$.
    \item Ensure that all vertices that have imbalance $\neq 0 \mod p$, actually have imbalance $+1$ or $-1$ (by splitting every such vertex into $p$ vertices, each getting at most one edge).
    \item Transform every solution that has imbalance $+1$ into $p-1$ solutions with imbalance $-1$ instead (by pointing to $p-1$ new vertices).
\end{enumerate}

From now on we assume that all \imba{p}-instances have this form. Given an instance $I$ of problem $\Pi$, let $(G_1^I,s_1^I)$ denote the first oracle query made by $M$ on input $I$, where $G_1^I$ is the \imba{p} graph (represented implicitly by circuits) and $s_1^I$ is the starting vertex. From now on we omit the superscript $I$ for better readability. For any solution $t_1$ to $(G_1,s_1)$, let $(G_2(t_1),s_2(t_1))$ be the second oracle query made by $M$, if the first query returned $t_1$. We construct a big graph $G$ that contains a copy of $G_1$ and a copy of $G_2(t_1)$ for each solution $t_1$ of $(G_1,s_1)$. A vertex $u$ in $G_2(t_1)$ is represented as $(t_1,u)$ in $G$. For each such $t_1$, we add an edge from $t_1$ to $(t_1,s_2(t_1))$. Note that these two vertices are now balanced. Thus, the instance $(G,s_1)$ has the following property: all solutions are of the form $(t_1,t_2)$, where $t_1$ is a solution to $(G,s_1)$, and $t_2$ is a solution to $(G_2(t_1),s_2(t_1))$. The straightforward generalisation of this construction for a polynomial number of queries (instead of 2), yields a graph $G$ such that any solution yields consistent query answers for a complete run of $M$ on input $I$. Thus, we obtain a Turing reduction that only needs to make one oracle query and then simulates $M$ with these query answers.

It remains to show that this graph $G$ can be constructed in polynomial time from $I$, i.e., we can efficiently construct circuits that compute the edges incident on any given node. This is easy to see, because any node contains enough information to simulate a run of $M$ up to the point that is needed to determine the neighbours in $G$. We omit the full details, since the formal arguments are analogous to the ones in the corresponding proofs in~\cite{buss2012propositional,johnson2011thesisNPsearch}.

\subsection*{Acknowledgements}

I would like to thank Aris Filos-Ratsikas and Paul Goldberg for helpful discussions, as well as an anonymous reviewer for suggestions that helped improve the presentation of the paper. This work was supported by an EPSRC doctoral studentship (Reference 1892947).

\bibliography{ppa-k}

\appendix

\section{Technical Lemmas for \texorpdfstring{\cref{thm:ppakl-char}}{Theorem~\ref*{thm:ppakl-char}}}\label{app:sec:char-lemmas}

The proof ideas from~\cite{beame1996counting} are used to construct some of these reductions.

\begin{lemma}\label{lem:ppa-krl_in_ppaklmodk}
Let $k \geq 2$ and $r \geq 1$. For all $0 < \ell < kr$ with $\ell \neq 0 \mod k$, it holds that $\ppakl{kr}{\ell} \subseteq \ppakl{k}{(\ell \mod k)}$.
\end{lemma}

\begin{proof}
Consider any instance of \bipal{kr}{\ell}. Split every vertex $v$ into $r$ versions $v_1, \dots, v_r$ and assign every incident edge to exactly one of these versions, e.g., by ordering the neighbours in increasing lexicographic order. Then exactly one version of the known starting vertex will have degree $\ell \mod k$ and all its other versions will have degree $0$ or $k$. In the new graph, any vertex that has degree not in $\{0,k\}$ (apart from this one version of the starting vertex) will immediately yield a solution of the original instance. Thus, we have reduced to an instance of \bipal{k}{(\ell \mod k)}.
\end{proof}

\begin{lemma}\label{lem:ppa-kl_in_ppa-krlr}
Let $k \geq 2$ and $r \geq 1$. For any $0 < \ell < k$ it holds that $\ppakl{k}{\ell} \subseteq \ppakl{kr}{\ell r}$.
\end{lemma}

\begin{proof}
Consider any instance of \bipal{k}{\ell}. Assign weight $r$ to every edge. Clearly, the starting vertex now has degree $\ell r$ and any other vertex with degree not in $\{0,kr\}$ yields a solution to the original instance. Finally, note that we can remove the weights without changing the degrees and adding any new solutions by using the ``mitosis'' gadgets.
\end{proof}

\begin{lemma}\label{lem:ppa-kll_in_ppa-k1}
Let $k \geq 2$ and $\ell \geq 1$. Then $\ppakl{k\ell}{\ell} \subseteq \ppakl{k}{1}$.
\end{lemma}

\begin{proof}
We reduce \groupl{k\ell}{\ell} to \groupl{k}{1} by adapting the proof in~\cite[Lemma 2.3]{beame1996counting} to obtain a reduction. Consider an instance of \groupl{k\ell}{\ell}, i.e., a partition of the set of integers $[m,2^n-1]$ given by a circuit $C$ for some $m$ such that $2^n-m = \ell \mod k\ell$. This means that there exists some integer $\alpha \in [0,2^{n-1}]$ such that $2^n-m=\ell+\alpha k \ell$. Clearly, there exists some integer $\beta \in [0,(\ell-1)!-1]$ such that $\alpha+\beta = 0 \mod (\ell-1)!$, which implies $2^n-m+\beta k \ell = \ell \mod (k \cdot \ell!)$ (if $\ell \leq 2$, then this holds with $\beta = 0$). Thus, if we add $\beta$ sets of size $k \ell$, the size of the ground set will be $= \ell \mod (k \cdot \ell!)$. Assuming that $n$ is large enough such that $2^n \geq k \cdot \ell! \geq \beta k \ell$, we can achieve this by letting $n'=n+1$, $m'=m+2^n-\beta k \ell$ and extending $C$ to also include the additional $\beta$ sets of size $k \ell$. It is easy to see that this can be done efficiently and will yield a partition of $[m',2^{n+1}-1]$ such that any mistake immediately yields a mistake in the original partition. Thus, we can assume without loss of generality that $2^n-m = \ell \mod (k \cdot \ell!)$.

Let $S$ denote the set of all subsets of $\{m,m+1, \dots, 2^n-1\}$ of size exactly $\ell$. Note that $|S|= \binom{2^n-m}{\ell} = \prod_{i=0}^{\ell-1} \frac{2^n-m-i}{\ell-i} = 1 \mod k$, since $2^n-m-i = \ell-i \mod k(\ell-i)$. We will now describe how to construct a partition of $S$ into sets of size $k$ such that any mistake yields a solution of the original instance.

Recall that we can assume that $C^{k\ell}(x)=x$ for all $x$. Thus, every $x$ yields an orbit $O(x)=\{x,C(x),C^2(x), \dots, C^{k\ell-1}(x)\}$ of size at most $k\ell$. In particular, we can pick the lexicographically smallest element of every orbit to be its representative. Denote by $R(x)$ the representative of the orbit containing $x$. We then have that $R(x)=R(y)$, if and only if $x$ and $y$ lie in the same orbit, i.e., in the same set in the original partition. For $a,b \in S$ we write $a \equiv b$ if $a$ and $b$ contain exactly the same number of elements from each set of the original partition. This can be checked efficiently by computing the representative of each element in $a$, ordering these lexicographically (with repetitions), doing the same for $b$ and checking if the two lists are identical. This is an equivalence relation and we denote the equivalence class of $a \in S$ under $\equiv$ by $[a]$.

For any $a \in S$ there exist distinct representatives $x_1, \dots, x_s$, $s \leq \ell$, and $\alpha_1, \dots, \alpha_s \geq 1$ with $\sum_{i=1}^s \alpha_i = \ell$ such that $a$ contains exactly $\alpha_i$ elements from the orbit represented by $x_i$, for all $i$. Thus, the size of $[a]$ is exactly $\prod_{i=1}^s \binom{k\ell}{\alpha_i}$, assuming that the orbits of $x_1,\dots,x_s$ all have size $k\ell$, i.e., do not yield a solution to the original problem. It was shown in the proof of~\cite[Lemma 2.3]{beame1996counting} that this quantity is a multiple of $k$. Thus, the equivalence class of $a$ can be perfectly partitioned into sets of size $k$. We now describe a way to do this explicitly and efficiently. Assume that the representatives $x_1,\dots, x_s$ are in increasing lexicographic order. Find the smallest index $i$ such that $k$ divides $\binom{k\ell}{\alpha_i}$. Let $F$ denote an arbitrary fixed efficient bijection between $\{0, \dots, \binom{k\ell}{\alpha_i}-1\}$ and $\binom{O(x_i)}{\alpha_i}$, where this denotes the set of all subsets of $O(x_i)$ of size exactly $\alpha_i$.

The circuit $C'$ determines the image of $a \in S$ by first computing $x_1, \dots, x_s$ and $\alpha_1, \dots, \alpha_s$ as described above, and determining the smallest index $i$ as explained above. Let $a_i = a \cap O(x_i)$. The circuit outputs
$$(a \setminus a_i) \cup F(\lfloor F^{-1}(a_i)/k \rfloor \cdot k +(F^{-1}(a_i)+1 \mod k)).$$
It is easy to check that as long as $|O(x_j)| = k \ell$ for all $j$, this procedure partitions $[a]$ into sets of size $k$. The last step is to set $m'=2^{n\ell}-|S|$ and construct an efficient bijection between $S$ and $\{2^{n\ell}-|S|,\dots,2^{n\ell}-1\}$, which is easy to do.
\end{proof}

\begin{lemma}\label{lem:ppa-k1_in_ppa-r1}
Let $k_1,k_2 \geq 2$. If all prime factors of $k_2$ also divide $k_1$, then $\ppakl{k_1}{1} \subseteq \ppakl{k_2}{1}$.
\end{lemma}

\begin{proof}
Similarly to our proof of \cref{lem:ppa-kll_in_ppa-k1}, we adapt the proof of the corresponding statement for the counting formulas from Beame et al.~\cite[Lemma 2.5]{beame1996counting} in order to obtain a polynomial-time reduction.

Since all prime factors of $k_2$ divide $k_1$, there exists some $\ell$ (bounded by a constant, e.g., $k_2$) such that $k_2$ divides $k_1^\ell$. From \cref{lem:ppa-krl_in_ppaklmodk} we know that $\ppakl{k_1^\ell}{1} \subseteq \ppakl{k_2}{1}$. Thus, it suffices to show that $\ppakl{k_1}{1} \subseteq \ppakl{k_1^\ell}{1}$. We write $k=k_1$ from now on.

Consider an instance $(C,m)$ of \groupl{k}{1}. Since $2^n-m = 1 \mod k$, it follows that there exists some $r$ such that $(2^n-m)^r = 1 \mod k^\ell$ by Euler's totient theorem. Pick such an $r \geq \ell$ -- any large enough multiple of the totient $\phi(k^\ell)$ will do -- and note that $r$ is bounded by some constant, since both $\ell$ and $\phi(k^\ell)$ are.

Assume without loss of generality that $C^k(x) = x$ for all $x$. We construct a partition of $\{m, \dots, 2^n-1\}^r$ explicitly as follows. For any $(x_1, \dots, x_r) \in \{m, \dots, 2^n-1\}^r$, let $i \in \{1, \dots \ell\}$ denote the largest index such that $x_j$ is the lexicographically smallest element in the orbit of $x_j$ under $C$ (i.e., $x_j$ is the representative of $O(x_j)$), for all $j < i$. If there is no such $x_j$, set $i=1$. If $i > \ell$, set $i=\ell$. The circuit $C'$ computes $C'(x_1, \dots, x_r) = (C(x_1),C(x_2), \dots, C(x_i), x_{i+1}, \dots x_r)$. It is easy to see that if the orbits $O(x_j)$ under $C$ all have size $k$, then this yields an orbit of size $k^\ell$. Thus, any orbit in the new instance that has size different from $k^\ell$ immediately yields some orbit of the original instance that does not have size $k$.

The final step is to set $m'=2^{nr}-(2^n-m)^r$ and construct an efficient bijection between $\{m', \dots, 2^{nr}-1\}$ and $\{m, \dots, 2^n-1\}^r$, which is easy to do.
\end{proof}

\end{document}